\documentclass{article}
\usepackage{hyperref,amssymb,amsmath,graphicx,verbatim,amsthm,fullpage, xcolor}

\newtheorem{theorem}{Theorem}
\newtheorem{cor}[theorem]{Corollary}
\newtheorem{lemma}[theorem]{Lemma}
\newtheorem{prop}[theorem]{Proposition}
\newtheorem{claim}{Claim}
\newtheorem*{claim*}{Claim}

\newtheorem*{remark*}{Remark}
\newtheorem{remark}[theorem]{Remark}

\newtheorem{conjecture}{Conjecture}
\newtheorem{question}{Question}

\newcommand{\C}{\mathbb{C}}

\newcommand{\N}{\mathbb{N}}

\newcommand{\R}{\mathbb{R}}

\newcommand{\rk}{\mathrm{rk}}
\renewcommand{\sc}{\sigma}

\newcommand\eps{\epsilon}

\renewcommand\mod{\mathrm{~mod}}

\begin{document}

\title{On families of anticommuting matrices}
\author{Pavel Hrube\v{s}\thanks{\texttt{pahrubes@gmail.com}. Supported by  ERC grant FEALORA 339691.}}

\maketitle

\begin{abstract} Let $e_{1},\dots, e_{k}$ be complex $n\times n$ matrices such that $e_{i}e_{j}=-e_{j}e_{i}$ whenever $i\not=j$. 
We conjecture that
\begin{itemize}
\item $\rk(e_{1}^{2})+\rk(e_{2}^{2})+\cdots+\rk(e_{k}^{2})\leq O(n\log n)$.
\end{itemize}
We show that
\begin{enumerate} \item\label{abs:1} $\rk(e_{1}^{n})+\rk(e_{2}^{n})+\cdots+\rk(e_{k}^{n})\leq O(n\log n)$,
\item\label{abs:2} if $e_{1}^{2},\dots, e_{k}^{2}\not=0$ then $k\leq O(n)$, 
\item\label{abs:3} if $e_{1},\dots, e_{k}$ have full rank, or at least $n-O(n/\log n)$, then $k=O(\log n)$.
\end{enumerate}
\ref{abs:1} implies that the conjecture holds
 if
$e_{1}^{2},\dots, e_{k}^{2}$ are diagonalizable (or if $e_{1},\dots,e_{k}$ are). \ref{abs:2} and \ref{abs:3} show it holds
when their rank is sufficiently large or sufficiently small. 
\end{abstract}

\section{Introduction}
Consider a family $e_{1},\dots, e_{k}$ of complex $n\times n$ matrices which  pairwise anticommute; i.e., $e_{i}e_{j}=-e_{j}e_{i}$ whenever $i\not=j$.
A standard example is a representation of a Clifford algebra, which gives an anticommuting family of $2\log_{2}n+1$ invertible matrices, if $n$ is a power of two (see Example 1).  This is known to be tight: if all the matrices $e_{1},\dots,e_{k}$ are invertible then $k$ is at most $2\log_{2}n+1$. 
(see \cite{Shapiro} and Theorem \ref{thm:intro1}). However, the situation is much less understood when the matrices are singular. As an example, take the following problem: 
\begin{question}\label{question} 
Assume that every $e_{i}$ has rank at least $2n/3$. Is $k$ at most $O(\log n)$? 
\end{question}
\noindent
 We expect the answer should be positive, though we can show only that $k\leq O(n)$. Such a problem can be solved under some extra assumptions. In \cite{clif}, it was shown that an anticommuting family of diagonalisable matrices can be ``decomposed'' into representations of Clifford algebras. This indeed answer Question \ref{question} if the $e_{i}$'s are diagonalisable. 
In this paper, we formulate a conjecture which relates the size of an anticommuting family with the rank of matrices in the family. We prove some partial results in this direction. In sum, the situation is clear when the matrices are diagonalisable, or their squares are diagonalisable, or even $\rk(e_{i}^{2})=\rk(e_{i}^{3})$. However, we can say very little about the case when the matrices are nilpotent.

One motivation for this study  is to understand sum-of-squares composition formulas. 
A sum-of-squares formula is an identity 
\begin{equation}
\label{Rsos}
(x_1^2+x_2^2+\cdots +x_k^2)\cdot (y_1^2+y_2^2+\cdots +y_k^2)= f_1^2+f_2^2+\cdots +f_n^2\,,
\end{equation}
where $f_{1},\dots,f_{n}$ are bilinear complex\footnote{The problem is often phrased over $\R$ when the bilinearity condition is automatic.} polynomials. We want to know how large must $n$ be in terms of $k$ so that such an identity exists. This problem has a very interesting history, and we refer the reader to the the monograph \cite{Shapiro} for details.  A classical result of Hurwitz \cite{Hurwitz} states that $n=k$ can be achieved 
only for $k \in \{1,2,4,8\}$. Hence, $n$ is strictly larger than $k$ for most values of $k$, but it is not known how much larger.
In particular, we do not known whether $n\geq \Omega(k^{1+\eps})$ for some $\eps>0$. In \cite{AmirAvi}, it was shown that such a lower bound would imply an exponential lower bound in a certain circuit model (while the authors obtained an $\Omega(n^{7/6})$ lower bound on \emph{integer} composition formulas in \cite{intercalate}). 
We point out that our conjecture about anticommuting families implies $n\geq \Omega(k^{2}/\log k)$, which would be tight. This connection is hardly surprising: already Hurwitz's theorem, as well as the more general Hurwitz-Radon theorem \cite{Hurwitz2,Radon}, can be proved by reduction to an anticommuting system. 

\section{The expected rank of anticommuting families}\label{sec:conjecture}

A family $e_{1},\dots,e_{k}$ of $n\times n$ complex matrices will be called \emph{anticommuting} if $e_{i}e_{j}=-e_{j}e_{i}$ holds for every \emph{distinct} $i,j\in \{1,\dots,k\}$. We conjecture that the following holds ($\rk(A)$ is the rank of the matrix $A$):

\begin{conjecture}\label{conjecture} Let $e_{1},\dots, e_{k}$ be an anticommuting family of $n\times n$ matrices. Then \[\sum_{i=1}^{k}\rk(e_{i}^{2})\leq O(n\log n)\,.\]
\end{conjecture}

The main motivation is the following theorem:
\begin{theorem}\label{thm:intro1}\cite{Shapiro} Let $e_{1},\dots, e_{k}$ be an anticommuting family of $n\times n$ invertible matrices. Then $k\leq 2\log_{2}n+1$.
The bound is achieved if $n$ is a power of two.
\end{theorem}
Under the assumption that $e_{i}^{2}$ are scalar diagonal matrices, this appears in \cite{Newman} (though it may have been known already to Hurwitz). As stated, it can be found in \cite{Shapiro} (Proposition 1.11 and Exercise 12, Chapter 1). There, an exact bound is given
\begin{equation}\label{eq:sha} k\leq 2q+1\,,\,\hbox{ if } n=m2^{q} \hbox{ with } m \hbox{ odd}\,.
\end{equation} 
Theorem \ref{thm:intro1} shows, first, that the Conjecture holds for invertible matrices and, second, that the purported upper bound cannot be improved: taking $2\log_2n+1$ full rank matrices gives $\sum \rk(e_{i}^{2})=(2\log_{2}+1)n$. 

A key aspect of Conjecture \ref{conjecture} is that  $\sum \rk(e_{i}^{2})$ is bounded in terms of a function of $n$ only. This would fail, had we counted $\sum\rk(e_{i})$ instead. For consider $2\times 2$ matrices
\[e_{i}=\left( \begin{array}{l r} 0 & a_{i}\\
0 & 0 \end{array}\right)\,,\,\,a_{i}\not=0\,.\]
They trivially anticommute (as  $e_{i}e_{j}=e_{j}e_{i}=0$), but $\sum_{i=1}^{k}\rk(e_{i})=k$, which can be arbitrarily large. However, we also have $e_{i}^{2}=0$ and this example is vacuous when counting $\sum\rk(e_{i}^{2})$. 
 The minimum requirement of the Conjecture is that every anticommuting family with non-zero squares is finite. This is indeed the case:
 
 \begin{theorem}\label{thm:intro2} Let $e_{1},\dots, e_{k}$ be an anticommuting family of $n\times n$ matrices with $e_{1}^{2},\dots,e_{k}^{2}\not=0$.
 Then $k\leq O(n)$
 \end{theorem}
\noindent
In Theorem \ref{thm:exact2}, we will show that $k\leq 2n-3$ if $n$ is sufficiently large, which is tight.
\begin{cor} $\sum_{i=1}^{k}\rk(e_{i}^{2})\leq O(n^{2})$
\end{cor}

We will also show:
\begin{theorem}\label{thm:intro3} Let $e_{1},\dots, e_{k}$ be an anticommuting family of $n\times n$ matrices. Then \[\sum_{i=1}^{k}\rk(e_{i}^{n})\leq (2\log_{2}n+1)n\,.\]
\end{theorem}
\noindent
This implies:
\begin{cor} Conjecture \ref{conjecture} holds whenever $\rk(e_{i}^{2})=\rk(e^{3}_{i})$ for every $e_{i}$ (this is guaranteed if $e_{i}^{2}$ is diagonalisable). 
\end{cor}
\noindent
Note that if already $e_{1},\dots,e_{k}$ are diagonalisable, we obtain $\sum_{i=1}^{k}\rk(e_{i})\leq (2\log_{2}n+1)n$.

We will also generalise Theorem \ref{thm:intro1}. In Theorem \ref{thm:main1}, we show that the assumption that $e_{i}$ have full rank can be replaced by the assumption that they have almost full rank. This, together with Theorem \ref{thm:intro2}. shows that Conjecture \ref{conjecture} holds if the $e_{i}^{2}$ have either rank close to $n$ or close to $\log n$.
Finally, note that the Conjecture implies positive answer to Question \ref{question}: if $\rk(e_{i})\leq 2n/3$ then $\rk(e_{i}^{2})\geq n/3$ and so we must have $k\leq O(\log n)$.

\paragraph{Notation and organisation}
$[k]:=\{1,\dots,k\}$. $\C^{n\times m}$ will denote the set of $n\times m$ complex matrices. For a matrix $A$, $\rk(A)$ is its the rank. Spectrum of a square matrix $A$, $\sc(A)$, is the set of its eigenvalues. $A$ is nilpotent if $A^{r}=0$ for some $r$ (or equivalently, $A^{n}=0$,  or $\sc(A)=\{0\}$). 

In Section \ref{sec:examples}, we give examples of anticommuting families. In Section \ref{sec:main}, we prove Theorems \ref{thm:intro2}, \ref{thm:intro3} and \ref{thm:main1}. In Section \ref{sec:exact}, we prove (\ref{eq:sha})  and determine the bound from Theorem \ref{thm:intro2} exactly. In Section \ref{sec:sos}, we outline the connection between our conjecture and the sums-of-squares problem.

We note that our results hold in any field of characteristic different from two.

\section{Examples of anticommuting families}\label{sec:examples}

We give two examples of anticommuting families. They achieve optimal parameters within its class. Example 1 gives the largest anticommuting family of invertible matrices (Theorem \ref{thm:intro1}), Example 2 the largest family of anticommuting matrices with non-zero squares if $n>4$ (Theorem \ref{thm:exact2}).

\paragraph{Example 1 - invertible matrices}
Suppose that $e_{1},\dots,e_{k} \in \C^{n\times n}$ are anticommuting matrices. Then the following is a family of $k+2$ anticommuting matrices of dimension $2n\times 2n$:
\begin{equation}\label{eq:?} \left( \begin{array}{l r} I_{n} & 0\\
0 & -I_{n} \end{array}\right)\,,\,\, \left( \begin{array}{l r} 0 & I_{n}\\
-I_{n} &0 \end{array}\right)\,,\,\, \left( \begin{array}{l r} 0 & e_{1}\\
e_{1} & 0 \end{array}\right)\,,\dots\,,\,\, \left( \begin{array}{l r} 0 & e_{k}\\
e_{k} & 0 \end{array}\right)\,.
\end{equation}
Starting with a single non-zero $1\times 1$ matrix, this construction can be applied iteratively to construct
a family of $2\log_{2}n+1$ anticommuting invertible $n\times n$ matrices whenever $n$ is a power of two.
Moreover, each matrix is diagonalizable.
If $n$ is not a power of two but rather of the form $m2^{q}$ with $m$ odd, we instead obtain
 $2q+1$ such matrices.

\paragraph{Example 2 - nilpotent matrices, plus one}
If $n\geq 2$, consider $n\times n$ matrices of the form 
\[e_{i}=\left( \begin{array}{l c r} 0 &u_{i}& 0\\
 &  & v_{i}^{t}\\ &&0 \end{array}\right)\,,\]
where $u_{i},v_{i}\in \C^{n-2}$ are row-vectors. Then 
\[e_{i}e_{j}=\left( \begin{array}{l c r}  0&0 & u_{i}v_{j}^{t}\\
 &  & 0\\ && 0\end{array}\right)\,, \]
and so $e_{i}e_{j}=-e_{j}e_{i}$ iff $u_{i}v_{j}^{t}=-u_{j}v_{i}^{t}$ and $e_{i}^{2}\not=0$ iff $u_{i}v_{i}^{t}\not =0$. 
Setting $r:=n-2$, it is easy to construct row vectors $u_{1},\dots,u_{2r}$, $v_{1},\dots, v_{2r}\in \C^{r}$ such that for every $i,j\in [2r]$
\[u_{i}v_{i}^{t}\not=0\,,\,\, u_{i}v_{j}^{t}=-u_{j}v_{i}^{t}\,\hbox{ if }i\not=j\,.\]
This gives an anticommutung family
\[e_{1},\dots, e_{2n-4} \in \C^{n\times n}\,,\]
where every $e_{i}$ is nilpotent but satisfies $e_{i}^{2}\not=0$.
Note that one can add one more matrix to the family: the diagonal matrix
\[e_{0}:=\left( \begin{array}{l c r} -1 && \\
 & I_{n-2} & \\ &&-1 \end{array}\right)\,.\]
This gives $2n-3$ anticommuting matrices with non-zero squares.

\section{Lower bounds on family size} \label{sec:main}

In this section, we prove our main theorems.  
A first observation to make is the following:
\begin{remark}\label{remark}
If $e_{1},\dots,e_{k}$ anticommute and $e_{1}^{2},\dots, e_{k}^{2}\not=0$ then they are linearly independent.
\end{remark} 
\noindent
To see this, assume that
$e_{1}=\sum_{j>1}^{k}a_{j}e_{j}$. Since $e_{1}$ anticommutes with every $e_{j},j>1$, we have 
 $e_{1}^{2}=e_{1}(\sum a_{j}e_{j})=-(\sum a_{j}e_{j})e_{1}= -e_{1}^{2}$ and hence $e_{1}^{2}=0$. 
 
 This means that $k\leq n^{2}$ if $e_{1},\dots,e_{k}\in \C^{n\times n}$. We first show that $k$ must actually be smaller. 
 
 \begin{theorem}\label{thm:main1}[Theorem \ref{thm:intro2} restated] Let $e_{1},\dots, e_{k}\in \C^{n\times n}$ be an anticommuting family with $e_{1}^{2},\dots. e_{k}^{2}\not=0$.
 Then $k\leq O(n)$
 \end{theorem}
 
 In Theorem \ref{thm:exact2}, we will see that the correct bound is $2n-3$ if $n$ is sufficiently large.
 
 \begin{proof} First, there exist row-vectors $u,v\in \C^{n}$ such that
 $u e_{i}^{2}v^{t}\not=0\in \C$  for every $i\in [k]$.
 This is because we can view $ue_{i}^{2}v^{t}$ as a polynomial in the $2n$-coordinates of $u$ and $v$. If $e_{i}^{2}\not=0$, the polynomial is non-trivial, and so a generic $u,v$ satisfies $ue_{i}^{2}v^{t}\not =0$ for every $i\in [k]$.
  
 Let us define the $k\times k$ matrix $M$ by
  \[M_{ij}:= \{ue_{i}e_{j}v^{t}\}_{i,j\in[k]}\,.\] 
 Then $\rk(M)\leq n$. This is because $M$ can be factored as $M=L\cdot R$, where $L$ is $k\times n$ matrix with $i$-th row equal to $ue_{i}$ and $R$ is $n\times k$ with $j$-th column equal to $e_{j}v^{t}$. On the other hand, we have $\rk(M)\geq k/2$. This is because $M_{ii}\not=0$ and, since $e_{i} e_{j}=-e_{j}e_{i}$, $M_{ij}=-M_{ji}$ whenever $j\not=i$. Hence $M+M^{t}$ is a diagonal matrix with non-zero entries on the diagonal, $\rk(M+M^{t})=k$ and so $\rk(M)\geq k/2$. This gives $k/2\leq \rk(M)\leq n$ and so $k\leq 2n$.  
 \end{proof}
 
 Remark \ref{remark} can be generalised. For $A=\{i_{1},\dots,i_{r}\}\subseteq [k]$ with $i_{1}<\dots<i_{r}$,  let $e_{A}$ be the matrix $e_{i_{1}}e_{i_{2}}\cdots e_{i_{r}}$. 
 
 \begin{lemma}\label{lem:V} Let $e_{1},\dots, e_{k}$ be anticommuting matrices. For $p\leq k$, assume that for every $A\subseteq \{1,\dots, k\}$ with $|A|\leq p$ we have $\prod_{i\in A}e_{i}^{2}\not=0$. Then the matrices $e_{A}$, with $|A|\leq p$ and $|A|$ even,  are linearly independent (similarly with odd $|A|$).  
 \end{lemma}
 
 \begin{proof} Suppose that we have a non-trivial linear combination $\sum_{A \hbox{ even}}a_{A}e_{A}=0$. 
 Let $A_{0}$ be a largest $A$ with $a_{A}\not=0$. We will show that $\prod_{i\in A_{0}}e_{i}^{2}=0$ holds. This implies the statement of the lemma for even $A$'s; the odd case is analogous.
The proof is based on the following observations. First, $e_{i}$ and $e_{j}^{2}$ always commute. Second, if $i\not\in A$ then $e_{i}e_{A}= (-1)^{|A|}e_{A}e_{i}$, i.e., $e_{A}$ and $e_{i}$ commute or anticommute depending on the parity of $|A|$. 
 
 Without loss of generality, assume that $A_{0}=\{1,\dots,q\}$.
 For $r\leq q$ and $z\in \N$ let $S_{r}(z):= \{A\subseteq \{r+1,\dots, k\}: |A|=z \mod 2\}$. 
 We will show that for every $0\leq r\leq q$,
 \begin{equation}\label{eq:1}
 e_{1}^{2}\cdots e_{r}^{2}\left(\sum_{A\in S_{r}(r)}a_{[r]\cup A}e_{A}   \right)=0\,.
 \end{equation}  
If $r=0$,  (\ref{eq:1}) is just the equality $\sum_{A \hbox{ even}}a_{A}e_{A}=0$. Assume (\ref{eq:1}) holds for some $r<q$, and we want to show it holds for $r+1$.   
 Collecting terms that contain $e_{r+1}$ and those that do not, (\ref{eq:1})
  can be rewritten as 
 where \begin{align*} 
 e_{1}^{2}\cdots e_{r}^{2}e_{r+1}\left(\sum_{A\in S_{r+1}(r+1)}a_{[r+1]\cup A}e_{A}\right)= - e_{1}^{2}\cdots e_{r}^{2}\left(\sum_{B\in S_{r+1} (r)}a_{[r]\cup B} e_{B}\right)\,.
  \end{align*}
  Let $f$ and $g$ be the left and right hand side of the last equality. 
Since $A$ range over sets of parity $(r+1)\mod2$ and $B$ over sets with parity $r\mod 2$,  we have $e_{r+1}f= (-1)^{r+1}fe_{r+1}$  and $e_{r+1}g=(-1)^{r}ge_{r+1}$. Since $f=g$, this gives $e_{r+1}f=-fe_{r+1}0$ and so $e_{r+1}f=0$. Hence, 
 \[e_{1}^{2}\cdots e_{r}^{2}e_{r+1}^{2}\sum_{A\in S_{r+1}(r+1)}a_{[r+1]\cup A}e_{A}\,,\]
 as required in (\ref{eq:1}). 
 Finally, if we set $r:=q$ in (\ref{eq:1}), we obtain $e_{1}^{2}\cdots e_{q}^{2}\cdot a_{A_{0}}= 0$ (recall that $A_{0}$ is maximal) and so $e_{1}^{2}\cdots e_{q}^{2}=0$, as required.   
 \end{proof}

 Part \ref{r2} of the following theorem is a generalisation of Theorem \ref{thm:intro1}. Note that  part \ref{r1} gives $k\leq O(\log n)$ whenever $r\geq n-O(n/\log n)$.  
 
 \begin{theorem}\label{thm:main2} Let $e_{1},\dots, e_{k}$ be anticommuting matrices in $\C^{n\times n}$ and $r:=\min_{i\in [k]}\rk(e_{i}^{2})$.
 \begin{enumerate}
 \item\label{r1} If $r> n(1-1/c)$ with $c\in \N$ then $k\leq cn^{2/c}$.
 \item\label{r2} If $r> n\left(1-\frac{1}{2(\log_{2}n+1)}\right)$ then $k\leq 2\log_{2}n+1$.
 \end{enumerate}
 \end{theorem}

 \begin{proof} \ref{r1}. By Sylvester's inequality, we have $\rk(\prod_{i\in A}e_{i}^{2})> n-|A|n/c$. Hence $\prod_{i\in A}e_{i}^{2}\not=0$ whenever $|A|\leq c$. By Lemma \ref{lem:V}, the matrices $e_{A}$, $A\subseteq [k]$, $|A|=c$, are linearly independent. Hence ${{k}\choose{c}}\leq n^{2}$ and the statement follows from the estimate ${{k}\choose{c}}\geq (k/c)^{c}$.
 
 In \ref{r2}, assume that $k> 2\log_{2}n+1$ and, without loss of generality, $k\leq 2\log_{2}n+2$. As above, we conclude $e_{1}^{2}\cdots e_{k}^{2}\not=0$. The lemma shows that the products $e_{A}$, with $|A|$ even, are linearly independent. This gives $2^{k-1}\leq n^{2}$ and so $k\leq 2\log_{2}n+1$, a contradiction.    
 \end{proof}

Before proving  Theorem \ref{thm:intro3}, we discuss general structure of anticommuting families. One way to obtain such a family is via a direct sum of simpler families. A family which cannot be so decomposed will be called irreducible. In Proposition \ref{prop:irred}, we will state some properties of irreducible families which allow to conclude the theorem. 

 If $A_{1}\in \C^{r_{1}\times r_{1}}$ and $A_{2}\in\C^{r_{2}\times r_{2}}$, let $A_{1}\oplus A_{2}$ be the $(r_{1}+r_{2})\times(r_{1}+r_{2})$ matrix
 \[A_{1}\oplus A_{2}= \left( \begin{array}{l r} A_{1} & 0\\
0 & A_{2} \end{array}\right)\,.\]
 A family $e_{1},\dots, e_{k}\in C^{n\times n}$ will be called \emph{reducible}, if there exists an invertible $V$ such that   
 \begin{equation}\label{eq:red}
 Ve_{i}V^{-1}=e_{i}(1)\oplus e_{i}(2)\,,\,\, i\in [k]
 \end{equation} 
 where $e_{1}(1),\dots, e_{k}(1)\in \C^{r_{1}\times r_{1}}$, $e_{1}(2),\dots, e_{k}(2)\in \C^{r_{2}\times r_{2}}$, with $0<r_{1}<n$ and $r_{1}+r_{2}=n$. 
If no such decomposition exists, the family will be called \emph{irreducible}. 
 
Note that the similarity transformation $Ve_{1}V^{-1},\dots, Ve_{k}V^{-1}$ preserves anticommutativity (and rank), and that $e_{1},\dots,e_{k}$ anticommutes iff both $e_{1}(1),\dots, e_{k}(1)$ and $e_{1}(2),\dots, e_{k}(2)$  do.

 \begin{lemma}\label{lem:normal} Let $A$ and $B$ be square matrices of the form 
 \[A=\left( \begin{array}{l r} A_{1} & 0\\
0 & A_{2} \end{array}\right)\,,\,\, B=\left( \begin{array}{l r} B_{1} & B_{3}\\
B_{4} & B_{2} \end{array}\right)\,, \]
 where $A_{1},B_{1}\in \C^{n\times n}$, $A_{2},B_{2}\in \C^{m\times m}$. If $AB=-BA$,  
 the following hold:
 \begin{enumerate}\item\label{normal:1} if there is no $\lambda$ such that $\lambda\in \sc(A_{1})$ and $-\lambda\in \sc(A_{2})$ then $B_{3}=0$ and $B_{4}=0$,
 \item\label{normal:2} if $\sc(A_{1})=\{\lambda_{1}\}$ and $\sc(A_{2})=\{\lambda_{2}\}$ for some $\lambda_{1},\lambda_{2}\not=0$ then $B_{1},B_{2}=0$. 
 \end{enumerate}
 \end{lemma}

 \begin{proof} We first note the folowing:
 
 \begin{claim*} Let $X\in C^{p\times p}$, $Y\in \C^{q\times q}$ and $Z\in \C^{p\times q}$ be such that $XZ=ZY$. If $\sc{(X)}\cap\sc{(Y)}=\emptyset$ then $Z=0$.
 \end{claim*}
 \begin{proof} Without loss of generality, we can assume that $Y$ is upper triangular with its eigenvalues $\lambda_{1},\dots, \lambda_{r}$ on the diagonal.
 Let $v_{1},\dots, v_{q}$ be the columns of $Z$, and assume that some $v_{i}$ is non-zero. Taking the first such $v_{i}$ gives $Xv_{i}=\lambda_{i}v_{i}$ -- contradiction with $\lambda_{i}\not\in \sc{(X)}$. 
 \end{proof}
 Anticommutativity of $A$ and $B$ gives $A_{1}B_{3}=-B_{3}A_{2}$ and $A_{2}B_{4}=-B_{4}A_{1}$. If $A_{1}, A_{2}$ satisfy the assumption of
  \ref{normal:1}, we have $\sc(A_{1})\cap\sc(-A_{2})=\emptyset$ and so $B_{3},B_{4}=0$ by the Claim. We also have $A_{1}B_{1}=-A_{1}B_{1}$. If $A_{1}$ is as in \ref{normal:2}, we have $\sc(A_{1})\cap\sc(-A_{1})=\emptyset$ and so $B_{1}=0$; similarly for $B_{2}$.
 \end{proof}
 
 Given $A$ in Jordan normal form, Lemma \ref{lem:V} determines block-structure of $B$. For example, if $A$ is block-diagonal
 \[A=\left( \begin{array}{l c c r} A_{1} & &&\\
  & A_{2}&& \\
  &&A_{3}&\\
  &&&A_{4}\end{array}\right),\]
where $\sc(A_{1})=\{1\},\sc(A_{2})=\{-1\}$, $\sc(A_{3})=\{0\}$ and $\sc(A_{4})=\{2\}$. Then  
 \[B=\left( \begin{array}{l c c r}  0&B_{1} &&\\
 B_{2} &0 && \\
  &&B_{3}&\\
  &&&0\end{array}\right).\]

 \begin{prop}\label{prop:irred} Let $e_{1},\dots, e_{k}\in \C^{n\times n}$ be an irreducible anticommuting family. Then every $e_{i}$ is either invertible or nilpotent. Moreover,
 \begin{enumerate}
 \item\label{irred:2} for every $e_{i}$, $\sc(e_{i})\subseteq\{\lambda_{i},-\lambda_{i}\}$ for some $\lambda_{i}\in \C$,
 \item\label{irred:3} if at least two of the matrices are invertible then $n$ is even and the multiplicity of $\lambda_{i}$ is exactly $n/2$ in an invertible $e_{i}$.
 \end{enumerate}
 \end{prop}
 
 \begin{proof} \ref{irred:2}. Assume that there is some $e_{i}$ with eigenvalues $\lambda_{1},\lambda_{2}$ with $\lambda_{1}\not=-\lambda_{2}$. After a suitable similarity transformation, we can assume that
 \[e_{i}=\left( \begin{array}{l r} e_{i}' & 0\\
0 & e_{i}'' \end{array}\right)\,,\]
where $e_{i}'\in \C^{r\times r}$ $e_{i}''\in \C^{(n-r)\times(n-r)}$ are such that $\sc(e_{i}')\subseteq\{\lambda_{1},-\lambda_{1}\}$ and $\sc(e_{i}'')\cap\{\lambda_{1},-\lambda_{1}\}=\emptyset$, for some $0<r<n$. Lemma \ref{lem:normal} part \ref{normal:1} gives that every $e_{j}$ is of the form
\[e_{j}=\left( \begin{array}{l r} e_{j}' & 0\\
0 & e_{j}'' \end{array}\right)\]
and hence the family is reducible.

\ref{irred:2} implies that every $e_{i}$ is either invertible or nilpotent.
For \ref{irred:3}, assume that $e_{i}$ is non-singular. By \ref{irred:2}, we have $\sc(e_{i})\subseteq \{\lambda_{i},-\lambda_{i}\}$ for some $\lambda_{i}\not=0$. Decompose $e_{i}$ as above, but with $\sc(e_{i}')=\{\lambda_{1}\}$ and $\sc(e_{i}'')=\{-\lambda_{i}\}$. Hence $r$ is the multiplicity of $\lambda_{i}$. The previous lemma part \ref{normal:2}
shows that every $e_{j}$, $j\not=i$, is of the form 
\[e_{j}=\left( \begin{array}{l r} 0& e_{j}'\\
e_{j}''& 0 \end{array}\right)\,,\]
where $e_{j}'$ is $r\times (n-r)$ and $e_{j}''$ is $(n-r)\times r$. Hence $e_{j}$ has rank at most $2r$ and also at most $2(n-r)$. If some $e_{j}$ is invertible, we must have $r=n/2$. 
 \end{proof}
 
 \begin{theorem}\label{thm:main3}[Theorem \ref{thm:intro3} restated] Let $e_{1},\dots, e_{k}\in \C^{n\times n}$ be an anticommutative family. Then $\sum_{i=1}^{k}\rk(e_{i}^{n})\leq (2\log_{2}n+1)n$. 
 \end{theorem}
 
 \begin{proof} Argue by induction on $n$. If $n=1$, the statement is clear. If $n>1$, assume first that the family is irreducible. 
  By Proposition \ref{prop:irred}, every $e_{i}$ is either invertible or nilpotent. If $e_{i}$ is nilpotent then $e_{i}^{n}=0$ and it contributes nothing to the rank. On the other hand, Theorem \ref{thm:intro1} asserts that there can be at most $2\log_{2}n+1$ anticommuting invertible matrices and so indeed
 $\sum_{i=1}^{k}\rk(e_{i}^{n})\leq (2\log_{2}n+1)n\,.$

 If the family is reducible, consider the decomposition in (\ref{eq:red}).
 By the inductive assumption, $\sum \rk(e_{i}(z)^{n})\leq \sum\rk(e_{i}(z)^{r_{z}})\leq (2\log_{2}r_{z}+1)r_{z}$ for both $z\in\{1,2\}$.
 Since $\rk(e_{i}^{n})=\rk(e_{i}(1)^{n})+\rk(e_{i}(2)^{n})$, we obtain
 \begin{align*}\sum_{i=1}^{k}\rk(e_{i}^{n})\leq&\sum_{i=1}^{k}\rk(e_{i}(1)^{r_{1}})+\sum_{i=1}^{k}\rk(e_{i}(2)^{r_{2}})\leq\\
 \leq& (2\log_{2}r_{1}+1)r_{1}+(2\log_{2}r_{2}+1)r_{2}\leq (2\log_{2}n+1)(r_{1}+r_{2})= 
\\
 = & (2\log_{2}n+1)n\,. \end{align*}
 \end{proof}

 \section{Some exact bounds}\label{sec:exact}
 For completeness, we now sketch a proof of (\ref{eq:sha}) from Section \ref{sec:conjecture}. We then prove the exact bound in Theorem \ref{thm:intro2}. 
 
 \begin{prop}\label{thm:exact1} Let $e_{1},\dots, e_{k}$ be an anticommutative family of invertible $n\times n$ matrices, where $n=m2^{q}$ with $m$ is odd.
 Then $k\leq 2q+1$.  
 \end{prop}
 The bound is achieved by Example 1
 
 \begin{proof}[Proof sketch] Argue by induction on $n$. If $n>1$, the non-trivial case is when the family is irreducible. If $k> 1$, we
 can assume that 
 \begin{equation}\label{eq:exact1}e_{1}=\left( \begin{array}{l r} e_{1}'&0\\
 0&e_{1}'' \end{array}\right)\,,\,\,e_{j}=\left( \begin{array}{l r} 0& e_{j}'\\
e_{j}''& 0 \end{array}\right)\,, \,\hbox{ if } j>1.
\end{equation}
 where $e_{i}',e_{i}''\in \C^{n/2\times n/2}$ are invertible. This is because, by Proposition \ref{prop:irred}, we can write  $e_{1}$ as in (\ref{eq:exact1}) with $\sc{(e_{1}')}=\{\lambda\}$, $\sc(e_{1}'')=\{-\lambda\}$, $\lambda\not =0$. Lemma \ref{lem:normal} part \ref{normal:2} gives that every 
$e_{j},j>1$ must indeed be of the form required in (\ref{eq:exact1}).
 If $e_{2},\dots, e_{k}$ anticommute then so do the $k-2$ matrices $e_{2}e_{3},e_{2}e_{4}, \dots, e_{2}e_{k}$.
 If $j>1$,  \[e_{2}e_{j}
=\left( \begin{array}{l r} e_{2}'e_{j}''& 0\\
0& e_{2}''e_{j}' \end{array}\right),\]
and so $e_{2}'e_{3}'',\dots, e_{2}'e_{k}''$ is a family of $k-2$ invertible anticommuting matrices in $\C^{n/2\times n/2}$. The inductive assumption gives $k-2\leq 2(q-1)+1$ and so $k\leq 2q+1$ as required. 
 \end{proof}

 For a natural number $n$, let $\alpha(n)$ denote the largest $k$ so that there exists an anticommuting family $e_{1},\dots, e_{k}\in \C^{n\times n}$ with $e_{1}^{2},\dots,e_{k}^{2}\not =0$.
 
 \begin{theorem}\label{thm:exact2}  
 \[\alpha(n)=\left\{ \begin{array}{l} 2n-1\,, \hbox{ if } n\in\{1,2\} \\ 2n-2\,, \hbox{ if } n\in\{3,4\}  \\ 2n -3\,,\hbox{ if }n> 4
 \end{array}\right.\]
 \end{theorem}
 
 The rest of this section is devoted to proving the theorem.
 
 \begin{lemma} If $n>1$, $\alpha(n)$ equals the maximum of the following quantities: a) $2n-3$, b) $\max_{0<r<n}(\alpha(r)+\alpha(n-r))$, c) $2+\alpha(n/2)$ (where we set $\alpha(n/2):=-1$ if $n$ is odd). 
 \end{lemma}
 
 \begin{proof} That $\alpha(n)$ is at least the maximum is seen as follows. $\alpha(n)\geq \hbox{a)}$ is Example 2. $\alpha(n)\geq 2+\alpha(n/2)$ is seen from (\ref{eq:?}) in Example 1. For b), suppose we have two anticommuting families $e_{1}(z),\dots, e_{k_{z}}(z)\in \C^{r_{z}\times r_{z}}$, $z\in \{1,2\}$. Then the following is an anticommuting family of $(r_{1}+r_{2})\times(r_{1}+r_{2})$ matrices: $e_{1}(1)\oplus 0,\dots, e_{k_{1}}\oplus 0, 0\oplus e_{1}(2),\dots, 0\oplus e_{k_{2}}(2)$ (with $0\in \C^{r_{1}\times r_{1}}, \C^{r_{2}\times r_{2}}$ respectively). 
 
 We now prove the opposite inequality. Let $e_{1},\dots, e_{k}\in \C^{n\times n}$ be an anticommuting family with $e_{1}^{2},\dots, e_{k}^{2}\not=0$.
 We first prove two claims.
 
 \begin{claim}\label{claim1} If all the $e_{i}$'s are nilpotent then $k\leq 2(n-2)$.
 \end{claim}
 \begin{proof} By a theorem of Jacobson \cite{Jacobson}, see also \cite{Radjavi}, a family of anticommuting nilpotent matrices is simultaneously upper triangularisable. So let assume that $e_{1},\dots,e_{k}$ are upper triangular with zero diagonal, and proceed as in the proof of Theorem \ref{thm:main1}. For $M$ as defined in the proof, it is enough to show that $\rk(M)\leq n-2$, which gives $k\leq 2(n-2)$.  If the $e_{i}$'s are upper triangular with zero diagonal, we can see that the first column of $L$ and the last row of $R$ are zero. This means $\rk(M)=\rk(LR)\leq n-2$.  
 \end{proof}
 
 \begin{claim}\label{claim2} If $e_{1},e_{2}$ are invertible then $k\leq 2+\alpha(n/2)$.
 \end{claim}
   \begin{proof} As in the proof of Theorem \ref{thm:exact1}, we can assume that the matrices have the form (\ref{eq:exact1}). 
  Note that $e_{2}',e_{2}''$ are invertible and $e_{2}'e_{3}'',\dots, e_{2}'e_{k}''$ is an anticommuting family of $k-2$ matrices in $\C^{n/2\times n/2}$.
  If we show that $(e_{2}'e_{j}'')^{2}\not=0$ for every $j\in\{3,\dots, k\}$, we obtain $k-2\leq \alpha(n/2)$ as required.
  
Let $j\in \{3,\dots,k\}$. Anticommutativity of $e_{2}$ and $e_{j}$ gives $e_{2}'e_{j}''=-e_{j}'e_{2}''$ and $e_{2}''e_{j}'=-e_{j}''e_{2}'$. Hence
  \begin{align*} (e_{2}'e_{j}'')^{2}=e_{2}'e_{j}''e_{2}'e_{j}''&=e_{2}'(e_{j}''e_{2}')e_{j}''= - e_{2}'e_{2}''e_{j}'e_{j}''\,,\\
  &=e_{2}'e_{j}''(e_{2}'e_{j}'')= -e_{2}'e_{j}''e_{j}'e_{2}''\,.
  \end{align*}  
    If $(e_{2}'e_{j}'')^{2}=0$, the first equality gives $e_{j}'e_{j}''=0$ and the second $e_{j}''e_{j}'=0$ (recall that $e_{2}', e_{2}''$ are invertible).
  But since $e_{j}^{2}=e_{j}'e_{j}''\oplus e_{j}''e_{j}'$, this gives $e_{j}^{2}=0$ -- contrary to the assumption $e_{j}^{2}\not=0$.
   \end{proof}
   
   To prove the Lemma, assume first that $e_{1},\dots, e_{k}$ is irreducible. Then the $e_{i}$'s are either invertible  or nilpotent. If there is at most one invertible $e_{i}$, Claim \ref{claim1} gives  $k-1\leq 2(n-2)$, as in a). If at least two $e_{i}$'s are invertible, Claim \ref{claim2} gives $k\leq 2+\alpha(n/2)$, as in b).   
 If the family is reducible, write it as in (\ref{eq:red}). For $z\in\{1,2\}$, let $A_{z}:=\{i\in [k]: e_{i}(z)^{2}\not=0\}$. Then $A_{1}\cup A_{2}=[k]$ and so $k\leq    
 \alpha(r_{1})+\alpha(r_{2})$, as in c).  
 \end{proof}
 
 \begin{proof}[Proof of Theorem \ref{thm:exact2}]
 Using the Lemma, it is easy to verify that the theorem holds for $n\leq 4$. If $n>4$, the lemma gives $\alpha(n)\geq 2n-3$ and it suffices to prove the opposite inequality. Assume that  $n$ is the smallest $n>4$ such that $\alpha(n)> 2n-3$. 
 This means that for every $n'<n$, $\alpha(n')= 2n'-\eps(n')$ where $\eps(n')=1$ if $n'\in\{1,2\}$ and $\eps(n')>1$ otherwise. 
 Then either $\alpha(r)+\alpha(n-r)> 2n-3$ for some $0<r<n$, or $2+\alpha(n/2)> 2n-3$. The first case is impossible: we have
 $\alpha(r)+\alpha(n-r)= 2n -\eps(r)-\eps(n-r)$. But $\eps(r)+\eps(n-r)< 3$ implies $r,(n-r)\in\{1,2\}$ and so $n\leq 4$. If $2+\alpha(n/2)> 2n-3$ we have $2+2(n/2)-2\eps(n/2)> 2n-3$ and so $n< 5-2\eps(n/2)\leq 3$.
 \end{proof}
 
 \section{Sum-of-squares formulas}\label{sec:sos}
 We now briefly discuss the sum-of-squares problem. Let $\sigma(k)$ be the smallest $n$ so that there exists a sum-of-squares formula as in (\ref{Rsos}) from the Introduction.
 The following can be found in Chapter 0 of \cite{Shapiro}:
 
 \begin{lemma} $\sigma(k)$ is the smallest $n$ such that there exists $k\times n$ matrices $A_{1},\dots A_{k}$ which satisfy
 \[A_{i}A_{i}^{t}= I_{k}\,,\,\, A_{i}A_{j}^{t}=-A_{j}A_{i}^{t}\,,\,\hbox{ if }i\not=j\,,\]
for every $i,j\in [k]$. 
 \end{lemma}
 
 The matrices from the lemma can be converted to anticommuting matrices, which provides a connection between the sum-of-squares problem and Conjecture \ref{conjecture}, as follows.
 
 \begin{prop} \begin{enumerate}
 \item\label{sos2} If $\sigma(k)=n$, there exists an anticommuting family $e_{1},\dots, e_{k}\in \C^{(n+2k)\times(n+2k)}$ such that $\rk(e_{1}^{2}),\dots, \rk(e_{k}^{2})= k$. (Moreover, we have $e_{1}^{2}=e_{2}^{2}\dots=e_{k}^{2}$ and $e_{1}^{3},\dots,e_{k}^{3}=0$.)
 \item Hence, Conjecture \ref{conjecture} implies $\sigma(k)=\Omega(k^{2}/\log k)$.
 \end{enumerate}
 \end{prop}
 \begin{proof} Take the $(2k+n)\times (2k+n)$ matrices (with $0\in \C^{k\times k}$)
 \[e_{i}:=\left( \begin{array}{c c c}0& A_{i} & 0\\
& & A_{i}^{t}\\ &&0 \end{array}\right),\, i\in [k]\,.\]
The matrices have the required  properties as seen from  
 \[e_{i}e_{j}=\left( \begin{array}{c c c}0&0  & A_{i}A_{j}^{t}\\
& & 0\\ &&0 \end{array}\right).\] 
 We have $\sum_{i=1}^{k}\rk(e_{i}^{2})=k^{2}$. As $2k+n\leq 3n$, the Conjecture gives $k^{2}=\sum_{i=1}^{k}\rk(e_{i}^{2})\leq O(3n\log(3 n))$ and so $n\geq \Omega(k^{2}/\log k)$.
 \end{proof}
 
 We can see that the matrices obtained in \ref{sos2} are nilpotent, which is exactly the case of Conjecture \ref{conjecture} we do not know how to handle.
 Finally, let us note that part \ref{sos2} is too generous if $\sigma(k)=k$. In this case, we can actually obtain $k-1$ invertible anticommuting matrices in $\C^{k\times k}$. Again following \cite{Shapiro}, let 
 \[e_{1}:= A_{1}A_{k}^{t}\,,\,\,e_{2}:=A_{2}A_{k}^{t}\,,\dots\,,\,\, e_{k-1}:=A_{k-1}A_{k}^{t}\,. \]
 They anticommute, as seen from $A_{i}A_{k}^{t}A_{j}A_{k}^{t}=- A_{i}A_{k}^{t}A_{k}A_{j}^{t}=-A_{i}A_{j}^{t}$ (note that $A_{k}A_{k}^{t}=I$ implies $A_{k}^{t}A_{k}=I$ for square matrices). This is one way how to obtain Hurwitz's $\{1,2,4,8\}$-theorem: if $\sigma(k)=k$, we have $k-1$ invertible anticommuting matrices in $\C^{k\times k}$. By Theorem \ref{thm:main1}, this gives $k-1\leq 2\log_{2}k+1$  and hence $k\leq 8$. Furthermore, the precise bound in (\ref{eq:sha}) rules out the $k$'s which are not a power of two. 
 
 \bibliographystyle{plain}	
\bibliography{myrefs}

\end{document}